\documentclass[sigplan]{acmart}
\renewcommand\footnotetextcopyrightpermission[1]{}
\acmConference[IWLS'20]{}{July 27-30, 2020}{Virtual Conference}

\usepackage{amssymb}
\usepackage{amsmath}
\usepackage{amsfonts}
\usepackage{graphicx}
\usepackage{array}
\usepackage{algorithm}
\usepackage{algorithmic}
\usepackage{multirow}
\usepackage{mathtools}
\usepackage{color}
\usepackage{balance}

\newtheorem{definition}{Definition}

\newtheorem{claim}{Claim}
\newtheorem{assumption}{Assumption}

\newcommand*\BitNeg{\ensuremath{\mathord{\sim}}}
\newcommand*\BitAnd{\mathbin{\&}}
\newcommand*\BitOr{\mathbin{|}}
\newcommand*\BitXor{\mathbin{\char`\^}}

\newlength{\pgmtab}  
\newenvironment{program}{\renewcommand{\baselinestretch}{1}%
\begin{tabbing}\hspace{0em}\=\hspace{0em}\=%
\hspace{\pgmtab}\=\hspace{\pgmtab}\=\hspace{\pgmtab}\=\hspace{\pgmtab}\=%
\hspace{\pgmtab}\=\hspace{\pgmtab}\=\hspace{\pgmtab}\=\hspace{\pgmtab}\=%
\hspace{\pgmtab}\=\hspace{\pgmtab}\=\hspace{\pgmtab}\=
\+\+\kill}{\end{tabbing}\renewcommand{\baselinestretch}{\intl}}
\newcommand {\BEGIN}{{\bf begin\ }}
\newcommand {\ELSE}{{\bf else\ }}
\newcommand {\IF}{{\bf if\ }}

\newcommand {\DO}{{\bf do\ }}

\newcommand {\RETURN}{\mbox{\bf return\ }}

\newcommand {\TRUE}{\textsc{true}}
\newcommand {\FALSE}{\textsc{false}}

\newcommand {\INPUT}{{\bf input}}
\newcommand {\OUTPUT}{{\bf output}}

\newcommand {\OR}{{\bf or\ }}

\newcommand {\INITIALIZE}{{\bf initialize\ }}
\newcommand {\FOREACH}{{\bf foreach\ }}

\newcommand {\NULL}{{\bf NULL}}

\newcolumntype{R}[1]{>{\raggedleft\let\\\tabularnewline}p{#1}}
\newcolumntype{C}[1]{>{\centering\let\\\tabularnewline}p{#1}}

\begin{document}

\title{Simulation-Guided Boolean Resubstitution}

\author{Siang-Yun Lee}
\affiliation{
  \institution{EPF Lausanne}
  \country{Switzerland}
}

\author{Heinz Riener}
\affiliation{
  \institution{EPF Lausanne}
  \country{Switzerland}
}

\author{Alan Mishchenko}
\affiliation{
  \institution{University of California, Berkeley}
  \country{United States}
}

\author{Robert K. Brayton}
\affiliation{
  \institution{University of California, Berkeley}
  \country{United States}
}

\author{Giovanni De Micheli}
\affiliation{
  \institution{EPF Lausanne}
  \country{Switzerland}
}

\begin{abstract}

This paper proposes a new logic optimization paradigm based on circuit simulation, which reduces the need for Boolean computations such as SAT-solving or constructing BDDs. The paper develops a Boolean resubstitution framework to demonstrate the effectiveness of the proposed approach. Methods to generate highly expressive simulation patterns are developed, and the generated patterns are used in resubstitution for efficient filtering of potential resubstitution candidates to reduce the need for SAT validation.  Experimental results show that improvements in circuit size reduction were achieved by up to $74\%$, compared to a state-of-the-art resubstitution algorithm.
\end{abstract}
\maketitle

\section{Introduction}\label{sec:intro}
Logic optimization and circuit minimization~\cite{Brayton90,DeMicheli94} play an important role in \textit{electronic design automation} (EDA).  As the size and complexity of digital circuits grow, the demand for scalable optimizations of multi-level logic networks grows. Boolean methods, such as Boolean decomposition and resubstitution~\cite{IWLS2006,Mishchenko2011}, often improve the quality of logic synthesis but at the cost of more runtime taken by having to solve Boolean problems using a \textit{satisfiability problem} (SAT) solver or a \textit{binary decision diagram} (BDD) package. Algebraic and other local-search methods, on the other hand, are much faster because they are based on structural analysis, circuit simulation, \textit{sum-of-product} (SOP) factorization, or other time-efficient computations. However, the reductions in area and delay cannot compete with those achieved by Boolean methods.
 
In this paper, we introduce a promising new paradigm, \emph{simulation-guided logic synthesis}, that leverages simulation to minimize the number of expensive NP-oracle calls to equivalence checkers or SAT solvers during synthesis.  The runtime gained can be used to further improve the synthesis quality, leading to faster place-and-route, as happened in one study~\cite{Mishchenko10}.  The underlying hypothesis about using simulation is that \textit{expressive} simulation patterns can be amassed for a logic network and used later as an efficient filter to avoid unnecessary formal checks.  The proposed paradigm is useful in  algorithms dominated by expensive Boolean computation and is particularly suitable for techniques such as (1)~computation of structural choices~\cite{Chatterjee06} to improve the quality of mapping, (2)~scalable combinational equivalence checking for large designs~\cite{Mishchenko06b}, and (3)~gate matching between several versions of the same network. The resulting patterns can also be used in \textit{automatic test pattern generation} (ATPG)~\cite{Roth1966} and in circuit reliability analysis~\cite{DAC2010}.
 

We demonstrate the advantages of simulation guidance by presenting a fast and efficient Boolean resubstitution framework, which iterates over the nodes and attempts to re-express their functions using other nodes in the network.  If updating a node's function makes other nodes in its fan-in cone dangling (i.e., having no fan-out), they can be deleted, resulting in the reduction of the network's size. For the special case of replacing a node directly with an existing node, it is equivalent to \textit{functional reduction} (FRAIG)~\cite{fraig}. 

The resubstitution engine can quickly identify and rule out most illegal resubstitution candidates by simply comparing simulation signatures.  SAT is used in two tasks: generating simulation patterns before synthesis, and validating resubstitutions found using simulation during synthesis. The runtime of the former task is not critical because the resulting patterns are reusable by different engines working on the same design, or by the same tool working on similar versions of the design. On the other hand, the runtime of SAT validation is reduced substantially, provided that the simulation patterns are expressive enough. 
To this end, we study what makes simulation patterns \emph{expressive} and profile different pattern generation strategies, including random simulation, stuck-at-value testing, observability checking, and combinations of these.  In the process of resubstitution, pre-computed simulation patterns can be refined further with the counter-examples generated by SAT.

The contributions of the paper are: (1) a simulation-based logic optimization framework, which separates the computation of expressive simulation patterns and their use to validate Boolean optimization choices; (2) methods to increase expressive power of simulation patterns, resulting in reduced runtime due to fewer SAT calls; (3) improvements to the computation of resubstitution functions, resulting in better quality of logic optimization. 
The experimental results show that our new engine allows implicit consideration of global satisfiability don't-cares, which achieves $41\%$ improvement in circuit size reduction compared to a state-of-the-art windowing-based resubstitution algorithm~\cite{IWLS2006}. The proposed framework also allows low-cost extension of the window sizes used, resulting in a total of $74\%$ improvement.

 


The rest of this paper is organized as follows: After preliminaries are given in Section~\ref{sec:preli} and related works introduced in Section~\ref{sec:relate}, we first describe the general algorithmic framework in Section~\ref{sec:framework}. Then, pattern generation methods are explained in Section~\ref{sec:patgen}, and simulation-guided resubstitution techniques are proposed in Section~\ref{sec:resub}. Finally, experimental results are given in Section~\ref{sec:exp}, and conclusions in Section~\ref{sec:concl}.

\section{Preliminaries}\label{sec:preli}
\subsection{Logic Network}
Logic networks are \textit{directed acyclic graphs} (DAGs) composed of \textit{logic gates} and realizing \textit{Boolean functions}, which are functions defined over the Boolean space $\mathbb{B} = \{0,1\}$, taking \textit{primary inputs} as inputs and presenting the function outputs at the \textit{primary outputs}. In this paper, we work with \textit{and-inverter graphs} (AIGs), but this framework can also be applied to other types of logic networks.

A \textit{gate} in a logic network usually computes a simple function of its \textit{fan-ins}, and passes the resulting value to all of its \textit{fan-outs}. In the case of AIGs, a gate is always an AND gate, and the inverters are represented by complemented wires with no cost (that is, they do not add to the circuit size). We also refer to a gate as a \textit{node}. The \textit{transitive fan-in} (TFI) or the \textit{transitive fan-out} (TFO) of a node $n$ is a set of nodes such that there is a path between $n$ and these nodes in the direction of fan-in or fan-out, respectively.

\subsection{Don't-Cares}\label{subsec:dc}
The functionality of a logic circuit can be specified using an input-output relation. Input assignments when the output value does not matter, are called \textit{external don't-cares} (EXDCs). These assignments can be utilized to optimize the circuit.

For a set of internal nodes of a circuit, there might be some value combinations that never appear at these nodes. For example, an AND gate $g_1$ and an OR gate $g_2$ sharing the same fan-ins can never have $g_1 = 1$ and $g_2 = 0$ at the same time. This combination is a \textit{satisfiability don't-care} (SDC) of a common TFO node of $g_1$ and $g_2$.

An input assignment $\Vec{x}$ is said to be \textit{un-observable} for node $n$ if the circuit outputs do not change when $n$ is replaced by its negation $\overline{n}$. Or conversely, $n$ is \textit{un-testable} under $\Vec{x}$. Because the function of $n$ under $\Vec{x}$ does not matter, the un-observable patterns, called \textit{observability don't-cares} (ODCs), can be used to optimize node $n$.

\subsection{Simulation}
A \textit{simulation pattern} is a set of values assigned to the primary inputs. Circuit simulation is done by visiting nodes in a topological order. The \textit{simulation signature} of a node is an ordered set of values produced at the node under each simulation pattern. When the number of patterns is more than one, bit-parallel operations can be used to speed up simulation substantially.

\subsection{Resubstitution}
For each \textit{root} node in the circuit that we want to replace, we first estimate its \textit{gain}, or the number of nodes that can be deleted after a successful resubstitution, using the size of the node's \textit{maximum fan-out free cone} (MFFC). A node $n$ is said to be in the MFFC of the root node $r$ if $n$ is in the TFI of $r$ and all path from $n$ to the primary outputs pass through $r$. Then, a set of \textit{divisors} is collected. A divisor is a node that can be used to express the function of the root node. It should not be in the TFO cone of the root, otherwise the resulting circuit would be cyclic. It should also not be in the MFFC because nodes in the MFFC may be removed after resubstitution. Nodes depending on primary inputs that are not in the TFI of the root node can also be filtered out from the set of potential divisors.

A \textit{resubstitution candidate} (also abbreviated as a \textit{candidate}) is either a divisor itself or a simple function, named the \textit{dependency function}, built with several divisors. In the latter case, the candidate is represented by the top-most node of the implementation, named the \textit{dependency circuit}, of the function. A \textit{resubstitution}, or simply \textit{substitution}, is a pair composed of a root node and a resubstitution candidate, and it is said to be \textit{legal} if replacing the root node with the candidate does not change the global input-output relation of the logic network. Otherwise, the resubstitution is said to be \textit{illegal}.


\section{Related Work}\label{sec:relate}
Research in Boolean resubstitution techniques dates back to the 1990s~\cite{Sato1991, Kravets1999}. In the 2000s, efforts were made to improve the scalability of BDD-based computations~\cite{Kravets2004} and to move away from BDDs to simulation and Boolean satisfiability (SAT)~\cite{Mishchenko2005}. The structural analysis (windowing) was introduced to  speed up the algorithm further~\cite{Mishchenko2011}. In \cite{Mishchenko2005}, the dependency function is computed by enumerating its onset and offset cubes using SAT. Random simulation is used for initial filtering of resubstitution candidates. In \cite{Mishchenko2011}, simulation is also used to find potential candidates, which are then checked by SAT solving. The dependency function is computed using interpolation~\cite{Craig1957}. Windowing is used to limit the search space and the SAT instance size, with the inner window as a working space, and the outer window as the scope for computing don't-cares.

A state-of-the-art Boolean resubstitution algorithm for AIGs using windowing was presented in \cite{IWLS2006}. It relies entirely on truth table computation, without any use of BDDs or SAT. The search for divisors is limited to a window near the root node.  The window inputs are computed as a size-limited reconvergence-driven cut. The node functions in the window are expressed in terms of the cut variables. The dependency function is not computed as a separate step after minimizing its support, as in \cite{Mishchenko2011}. Instead, simple functions up to three AND gates are tried for resubstitution using several heuristic filters. The windowing-based resubstitution framework has been generalized to many different gates types including majority gates~\cite{Riener18} and complex gates~\cite{Amaru18}.

Random simulation is a core tool in logic synthesis and has been used successfully to reduce the runtime of various computations.  Functional reduction~\cite{fraig}, for example, uses random and guided simulation to identify equivalent nodes and merge them.  Effective combinational equivalence checking~\cite{Mishchenko06b} also is used to find cut-points between two networks that serve as stepping stones for the final proof of equivalence at the primary outputs. Motivated by the efficacy of these techniques, the simulation-guided paradigm in this paper focuses on identifying a set of expressive simulation patterns. Once identified, these can be reused multiple times to speed up logic synthesis for the same network in various applications.


\section{Framework}\label{sec:framework}
We introduce a promising new paradigm for logic synthesis that exploits fast bit-parallel simulation to reduce the number of expensive NP-hard checks, such as those based on SAT.  The rationale behind the idea is to pre-compute a set of ``expressive'' simulation patterns for a given logic network, which can rule out illegal transformations by comparing simulation signatures.

\begin{definition}\label{def:expressive}
We call a non-exhaustive set of simulation patterns \emph{expressive} for a logic network if the set can be used to pair-wise distinguish functionally non-equivalent nodes that either already exist in the logic network or can be derived from the existing nodes.
\end{definition}

Obviously, the  set of \textit{all} simulation patterns of  primary inputs satisfies this definition, but this  is typically too large for logic networks with $16$ or more primary inputs.  In practice, only  expressive simulation patterns that can be efficiently stored and simulated using less than, say, a few hundred or thousand bits are of interest.

\begin{assumption}\label{assumption}
We assume that, for a logic network with $N$ nodes, a set $S$ of expressive simulation patterns with size $|S| \leq C \cdot N$ exists, where $C$ is some constant parameter determined by the structure of the logic network.
\end{assumption}

This  means that  a set of expressive simulation patterns can be pre-computed, stored, and re-used by different logic synthesis engines when applied to the same design, or by the same engine when invoked multiple times. In the rest of this paper, this paradigm is demonstrated using Boolean resubstitution.

The resubstitution framework performs these steps:
\begin{enumerate}
    \item Generation of a set of expressive simulation patterns. In general, we can start with a set of random patterns, and refine or expand it with the techniques proposed in Section~\ref{sec:patgen}.
    \item Simulation of the network with these patterns to obtain simulation signatures for each node.
    \item Choosing a root node to be substituted. Estimating the gain by computing its MFFC and collecting the divisors. Skipping this node if the gain is too small or if there are no divisors.
    \item Searching for resubstitution candidates and the dependency function using simulation signatures. Details of this step are described in Section~\ref{sec:resub}.
    \item Validating the resubstitution with SAT solving by assuming non-equivalence. An UNSAT result validates the resubstitution, while a SAT result provides an input assignment under which the substituted network is non-equivalent to the original network. In the latter case, the counter-example is added to the set of simulation patterns. 
    \item Iterating starting from Step~$3$, until all nodes in the circuit have been processed.
\end{enumerate}


\section{Simulation Pattern Generation}\label{sec:patgen}
It was observed that expressive simulation patterns cannot be  derived directly from the input-output function of the logic network, but must account for some structural information.  An intuitive explanation of this may be that an input-output function can be implemented by a large number of structurally different logic networks.  This agrees with the idea of re-using simulation patterns in multiple optimization passes because the initial structure of the network often is determined by high-level synthesis and later carefully fine-tuned by logic optimization.  Consequently, only a small fraction of closely-related structures is encountered during logic optimization of the network.


Motivated by Assumption~\ref{assumption}, we suggest two simple orthogonal strategies for pre-computing simulation patterns:
\begin{enumerate}
\item Random patterns: this  generates random values for the primary inputs with equal probability of  $0$ or $1$ for each bit.
\item Stuck-at patterns: this  iteratively selects nodes and computes patterns that distinguish each from constant functions $0$ and $1$. 
\end{enumerate}
The implementation of the first strategy is straightforward. We describe the second one in the following section. Then, in Section~\ref{subsec:obs}, we propose an observability-based method to strengthen the computed stuck-at patterns.

\subsection{Stuck-at Values}\label{subsec:stuck-at}
Some nodes in the circuit may rarely produce a value ($0$ or $1$) during random simulation. For example, the output of an AND gate with many fan-ins may be $0$ most of the time, hence $1$ is rather rare and may be critical. Thus we refine the set of simulation patterns by checking that every node has both values appearing in its simulation signature. If only one value occurs, a new simulation pattern is created by solving a SAT problem, which forces the node to have the other value.

The algorithm is illustrated in Figure~\ref{fig:patgen-stuckat}. In line \texttt{01}, we can either start with an empty set or a random set of simulation patterns. Then, in line \texttt{04}, for each node in the circuit, if  $0$ or $1$ does not appear, we try to generate a pattern to express the missing value (lines \texttt{05-08}). In an un-optimized circuit, there may be nodes which never take one of the values, so these are replaced by a constant node in line \texttt{10}.

We can  strengthen the pattern set further by assuring both values appear multiple times (for example, at least $10$ times) in the signature of every node. This can be done by running the SAT solver multiple times while making sure it takes different computation paths.

\begin{figure}[h]
\mbox{}\hrulefill
\small
\begin{program}
 {\bf \textit{StuckAtCheck}}\\
 \INPUT: a circuit $C$ \\
 \OUTPUT: a set of expressive simulation patterns $S$ \\
 01 \> \> $S$ := a small set of random patterns; $C$.\textit{simulate}($S$)\\
 02 \> \> \INITIALIZE \textit{Solver}; \textit{Solver}.\textit{generate\_CNF}($C$)\\
 03 \> \> \FOREACH node $n$ in $C$ \DO\\
 04 \> \> \> \IF $n$.signature $= \Vec{0}$ \OR $n$.signature $= \Vec{1}$ \DO\\
 05 \> \> \> \> \IF $n$.signature $= \Vec{0}$ \DO \textit{Solver.add\_assumption}($n$)\\
 06 \> \> \> \> \ELSE \DO \textit{Solver.add\_assumption}($\neg n$)\\
 07 \> \> \> \> \IF \textit{Solver.solve}() $=$ SAT \DO\\
 08 \> \> \> \> \> $S$ := $S~\cup$ \{\textit{Solver.pi\_values}\}\\
 09 \> \> \> \> \ELSE \DO \\
 10 \> \> \> \> \> Replace $n$ with constant node.\\
 11 \> \> \RETURN $S$ 
\end{program}
\hrulefill
\vspace{-1em}
\caption{Generation of expressive simulation pattern by asserting stuck-at values.}
\label{fig:patgen-stuckat}
\vspace{-0.3cm}
\end{figure}

\subsection{Observability}\label{subsec:obs}
As described in Section~\ref{subsec:dc}, there may be some simulation patterns that are not observable with respect to an internal node; these patterns are deemed less expressive. Here, two cases are identified where a (re-)generation of an observable pattern may be done:
\begin{itemize}
    \item Case $1$: In $StuckAtCheck$  when a node is stuck at a value, and a new pattern is generated to express the other value, this pattern is not observable.
    \item Case $2$: A node assumes both values, but for all the patterns under which the node assumes one of the values, it is not observable.
\end{itemize}

To \textit{resolve} un-observable patterns, a procedure \textit{ObservablePatternGeneration} is devised, which generates an observable simulation pattern $\Vec{x}$ with respect to a given node $n$ and makes sure that $n$ expresses a specified value $v$ under $\Vec{x}$. This procedure builds a CNF instance, as shown in Figure~\ref{fig:tfo-miter}, and solves it using the SAT solver. If the instance is satisfiable, an observable pattern is generated (Claim~\ref{thm:observable}), and we say that the originally un-observable pattern is \textit{resolved}. Else if the solver returns UNSAT, we conclude that value $v$ at node $n$ is not observable. Hence, it can be replaced by the constant node in the respective polarity (Claim~\ref{thm:untestable}).

\begin{figure}[h]
\centering
\includegraphics[width=0.2\textwidth]{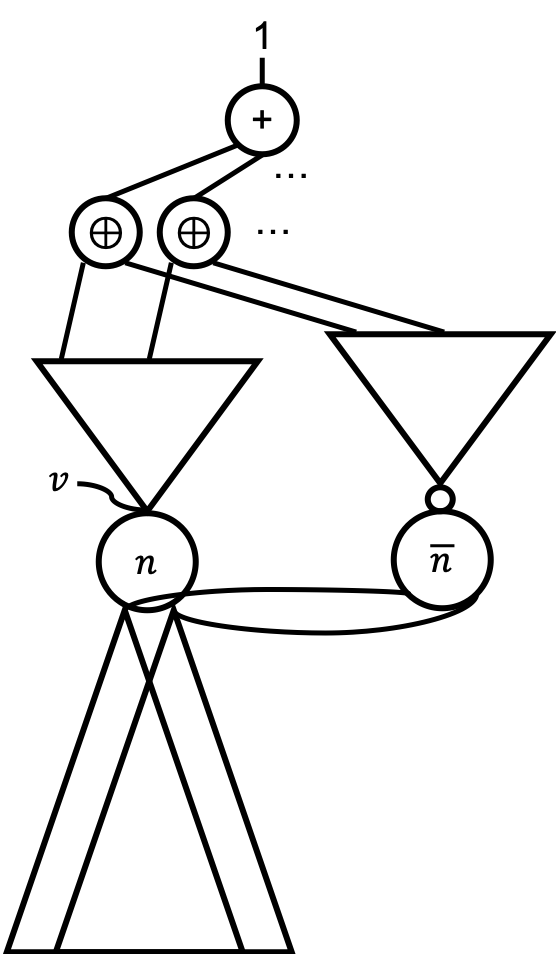}
\caption{Circuit of the CNF instance built in procedure \textit{ObservablePatternGeneration}. It is constructed by duplicating the TFO cone of $n$ and connecting it to $\overline{n}$. Primary outputs of the two TFO cones are matched and connected to XOR gates ($\oplus$), and the XOR gates are fed to an OR gate ($+$), whose output is asserted to be 1, forming a miter sub-circuit.}\label{fig:tfo-miter}
\end{figure}

\begin{claim}\label{thm:observable}
A satisfying input assignment $\Vec{x}$ in the circuit of Figure~\ref{fig:tfo-miter} is an observable pattern with respect to node $n$.
\end{claim}
\begin{proof}
By the definition in Section~\ref{subsec:dc}, $\Vec{x}$ is observable with respect to $n$ if the value of at least one of the primary outputs of the circuit under $\Vec{x}$ is different when $n$ is replaced by $\overline{n}$. This condition is ensured by the miter of the TFO cones of $n$ and $\overline{n}$ in Figure~\ref{fig:tfo-miter}.
\end{proof}

\begin{claim}\label{thm:untestable}
If a node $n$ is never observable with value $v$ ($v \in \{0,1\}$), then it can be replaced by constant $\neg v$ ($\neg0=1, \neg1=0$) without changing the circuit function(s). That is, there does not exist a primary input assignment $\Vec{x}$, such that one of the primary outputs has different values in the original circuit and in the substituted circuit.
\end{claim}
\begin{proof}
Assume the opposite: there exists a primary input assignment $\Vec{x}$, such that at least one of the primary outputs has a different value after substituting $n$ with $\neg v$. If the value of $n$ is $\neg v$ under $\Vec{x}$, all node values in the circuit, including primary outputs, remain unchanged if $n$ is replaced by $\neg v$. If the value of $n$ is $v$ under $\Vec{x}$, because $n$ is not observable with $v$, all primary outputs remain at the same value when the node value of $n$ changes to $\overline{n} = \neg v$, which contradicts the assumption.
\end{proof}

In order to limit the computation in large circuits, the TFO in Figure~\ref{fig:tfo-miter} can be restricted to nodes within a certain distance from $n$, called the \textit{depth} of the TFO cone, instead of extending all the way to primary outputs. In this case, all the leaves of the cone should be XOR-ed with their counterparts to build the miter. Using depth of $5$ is empirically a good tradeoff between quality and runtime.

After an observable pattern $\Vec{x}$ is generated, in Case~$1$, we can replace the pattern generated by \textit{StuckAtCheck} with $\Vec{x}$. In Case~$2$, we simply add $\Vec{x}$ to the set of patterns.



\section{Simulation-Guided Resubstitution}\label{sec:resub}
When the set of expressive simulation patterns is available, we exploit the simulation signatures of the nodes to implement Boolean resubstitution. The main difference of our algorithm, compared to a state-of-the-art resubstitution algorithm~\cite{Mishchenko2006}, is in the representation of divisors. Instead of using the complete truth table of the \textit{local} function of the node, we use the simulation signature approximating the \textit{global} function of the node.

The resubstitution algorithm is shown in Figure~\ref{fig:resub-basic}. Bit operators ($\BitNeg, \BitOr$ and $\BitAnd$) are implemented using bit-wise operations on simulation signatures.  Symbols $\neg, \wedge$ and $\vee$ indicate the creation of a complemented wire, an AND gate, and an OR gate (AND gate with the complemented inputs/output), respectively. 

In line \texttt{03}, procedure \textit{collect\_divisors} collects potential divisors in the TFI cone of $n$, excluding nodes in its MFFC or nodes depending entirely on other divisors. This computation can be extended to the whole circuit, excluding only  nodes in the TFO of $n$ or in its MFFC. In practice to keep the runtime reasonable, the number of collected divisors is limited.

First, resubstitution is tried without new nodes, as shown in lines \texttt{05-09}. If $n$ can be  expressed directly using a divisor or its negation, the network size is reduced by removing $n$ and its MFFC. Procedure \textit{verify} uses the SAT solver to try to find a pattern, under which nodes $n_1$ and $n_2$ have different values. The resubstitution is validated if the solver returns UNSAT (lines \texttt{25-26}); otherwise, a counter-example is added to the set of simulation patterns (lines \texttt{27-29}).

If the MFFC is not empty, substituting $n$ with simple functions of two divisors can be tried, which will lead to creating one new node. In the case of an AIG, the divisors  are partitioned into positive ($P$), negative ($N$) unate divisors and other, by checking the implication relation between their signatures (lines \texttt{11-16}). Then, in lines \texttt{17-22}, the signatures are compared to find potential resubstitutions using OR and AND functions. Resubstitution candidates of this type also need to be formally verified.

\begin{figure}[ht]
\mbox{}\hrulefill
\small
\begin{program}
 {\bf \textit{SimResub}}\\
 \INPUT: a root node $n$ in a simulated circuit $C$ \\
 \OUTPUT: a legal (verified) candidate to substitute $n$ \\
 01 \> \> \INITIALIZE \textit{Solver}; \textit{Solver}.\textit{generate\_CNF}($C$)\\
 02 \> \> $MFFC\_size$ := $|$\textit{compute\_MFFC}($n$)$|$\\
 03 \> \> $D$ := \textit{collect\_divisors}($n$)\\
 04 \> \> \IF $D = \emptyset$ \DO \RETURN \NULL\\
 05 \> \> \FOREACH divisor $d$ in $D$ \DO \texttt{/* resub-0 */}\\
 06 \> \> \> \IF $n$.signature $= d$.signature \DO\\
 07 \> \> \> \> \IF \textit{Solver.verify}($n$, $d$) \DO \RETURN $d$\\
 08 \> \> \> \IF $n$.signature $= \BitNeg d$.signature \DO\\
 09 \> \> \> \> \IF \textit{Solver.verify}($n$, $\neg d$) \DO \RETURN $\neg d$\\
 10 \> \> \IF $MFFC\_size = 0$ \DO \RETURN \NULL\\
 11 \> \> $P$ := $\emptyset$; $N$ := $\emptyset$\\
 12 \> \> \FOREACH divisor $d$ in $D$ \DO\\
 13 \> \> \> \IF $d$.signature $\rightarrow n$.signature \DO $P$ := $P \cup \{d\}$\\
 14 \> \> \> \ELSE \IF $\BitNeg d$.signature $\rightarrow n$.signature \DO $P$ := $P \cup \{\neg d\}$\\
 15 \> \> \> \ELSE \IF $n$.signature $\rightarrow d$.signature \DO $N$ := $N \cup \{d\}$\\
 16 \> \> \> \ELSE \IF $n$.signature $\rightarrow \BitNeg d$.signature \DO $N$ := $N \cup \{\neg d\}$\\
 17 \> \> \FOREACH pair of divisors $d_1,d_2$ in $P$ \DO \texttt{/* resub-1 */}\\
 18 \> \> \> \IF $n$.signature $= d_1$.signature $\BitOr~ d_2$.signature \DO\\
 19 \> \> \> \> \IF \textit{Solver.verify}($n$, $d_1 \vee d_2$) \DO \RETURN $d_1 \vee d_2$\\
 20 \> \> \FOREACH pair of divisors $d_1,d_2$ in $N$ \DO \texttt{/* resub-1 */}\\
 21 \> \> \> \IF $n$.signature $= d_1$.signature $\BitAnd~ d_2$.signature \DO\\
 22 \> \> \> \> \IF \textit{Solver.verify}($n$, $d_1 \wedge d_2$) \DO \RETURN $d_1 \wedge d_2$\\
 23 \> \> \RETURN \NULL \\ \\
 {\bf \textit{Solver.Verify}}\\
 \INPUT: two nodes, $n_1$ and $n_2$, in a simulated circuit $C$ \\
 \OUTPUT: whether it is legal to substitute $n_1$ with $n_2$ \\
 24 \> \> \textit{Solver.add\_assumption}(~\textit{literal}($n_1$) $\oplus$~ \textit{literal}($n_2$)~)\\
 25 \> \> \IF \textit{Solver.solve}() $=$ UNSAT \DO \\
 26 \> \> \> \RETURN \TRUE \\
 27 \> \> \ELSE \\
 28 \> \> \> $C$.\textit{add\_pattern}(\textit{Solver.pi\_values})\\
 29 \> \> \> \RETURN \FALSE
\end{program}
\hrulefill
\vspace{-1em}
\caption{Simulation-guided resubstitution.}
\label{fig:resub-basic}
\vspace{-0.3cm}
\end{figure}

\section{Experimental Results}\label{sec:exp}
The pattern generation and the simulation-guided resubstitution framework are implemented in C++-17 as part of the EPFL logic synthesis library \textit{mockturtle}\footnote{\hyperlink{github.com/lsils/mockturtle}{github.com/lsils/mockturtle}}. The experiments are performed on a Linux machine with Xeon 2.5 GHz CPU and 256 GB RAM. The OpenCore designs from IWLS'05 benchmark\footnote{\hyperlink{iwls.org/iwls2005/benchmarks.html}{iwls.org/iwls2005/benchmarks.html}} are used for testing.

In this section, we investigate the expressiveness of simulation patterns generated using different methods and compare their impact on resubstitution. The advantages of the framework are measured in terms of the circuit size reduction and compared against state-of-the-art. Also, quality/speed trade-offs are explored.

\subsection{Size of Simulation Pattern Set}

Intuitively, the more simulation patterns used, the higher is the chance that the framework saves time by not attempting to validate illegal resubstitutions, i.e. a larger set of simulation patterns is expected to be more expressive. Following Definition~\ref{def:expressive} in Section~\ref{sec:framework}, we measure the \textit{expressive power} of a pattern set using the percentage decrease in the number of counter-examples encountered in the resubstitution flow, compared to the baseline set calculated separately for each benchmark. Different from a typical resubstitution flow, the counter-examples are not added to the simulation set, to isolate the impact of the provided patterns.

We start by investigating the expressive power of random patterns based on their count. In Figure~\ref{fig:exp-randpat}, each bar represents how expressive is a pattern set of the respective size, compared to the baseline of using only four simulation patterns. The smaller sets are subsets of the larger sets to avoid the biasing effect of randomness. As the size grows by the factor of four (leading to 4, 16, 64, etc patterns), the expressive power increases very fast at first, as expected, but saturates at a few hundreds to a few thousands of patterns. Fortunately, a thousand patterns is still a practical size, for which bit-parallel simulation runs fast.

\begin{figure}[ht]
\centering
\includegraphics[width=0.48\textwidth]{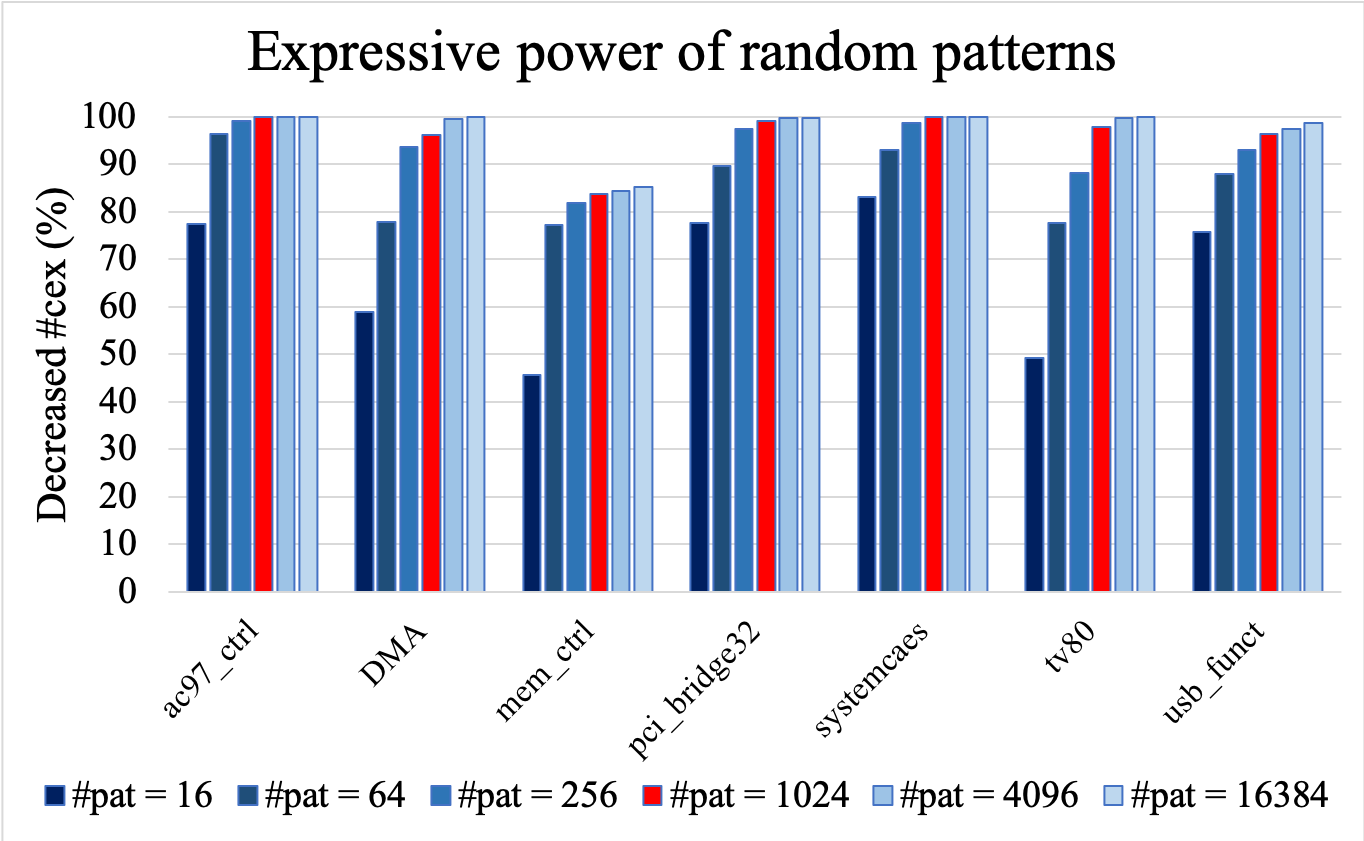}
\caption{Decreased percentages of counter-examples when provided with different number ($\#pat$) of random simulation patterns comparing to $\#pat = 4$.}\label{fig:exp-randpat}
\end{figure}

A similar phenomenon is observed when patterns are generated by $StuckAtCheck$. As discussed in Section~\ref{subsec:stuck-at}, additional patterns can be used to ensure that every node has at least $b$ bits of $0$ and $b$ bits of $1$ in its signature. In the following experiments, stuck-at patterns are abbreviated as ``\texttt{s-a}'', with a prefix ``\texttt{$b$x}'' listing parameter $b$. Since the stuck-at pattern counts are different for each testcase, the pattern set size is normalized to the circuit size and plotted in the logarithmic scale in Figure~\ref{fig:exp-sa} and the following figures.

In Figure~\ref{fig:exp-sa}, it is observed that larger sets of patterns are usually more expressive. Note that randomness plays a role in this case, since the default variable polarities, which determine initial variable values in the SAT solver, are randomly reset before each run.

\begin{figure}[ht]
\centering
\includegraphics[width=0.48\textwidth]{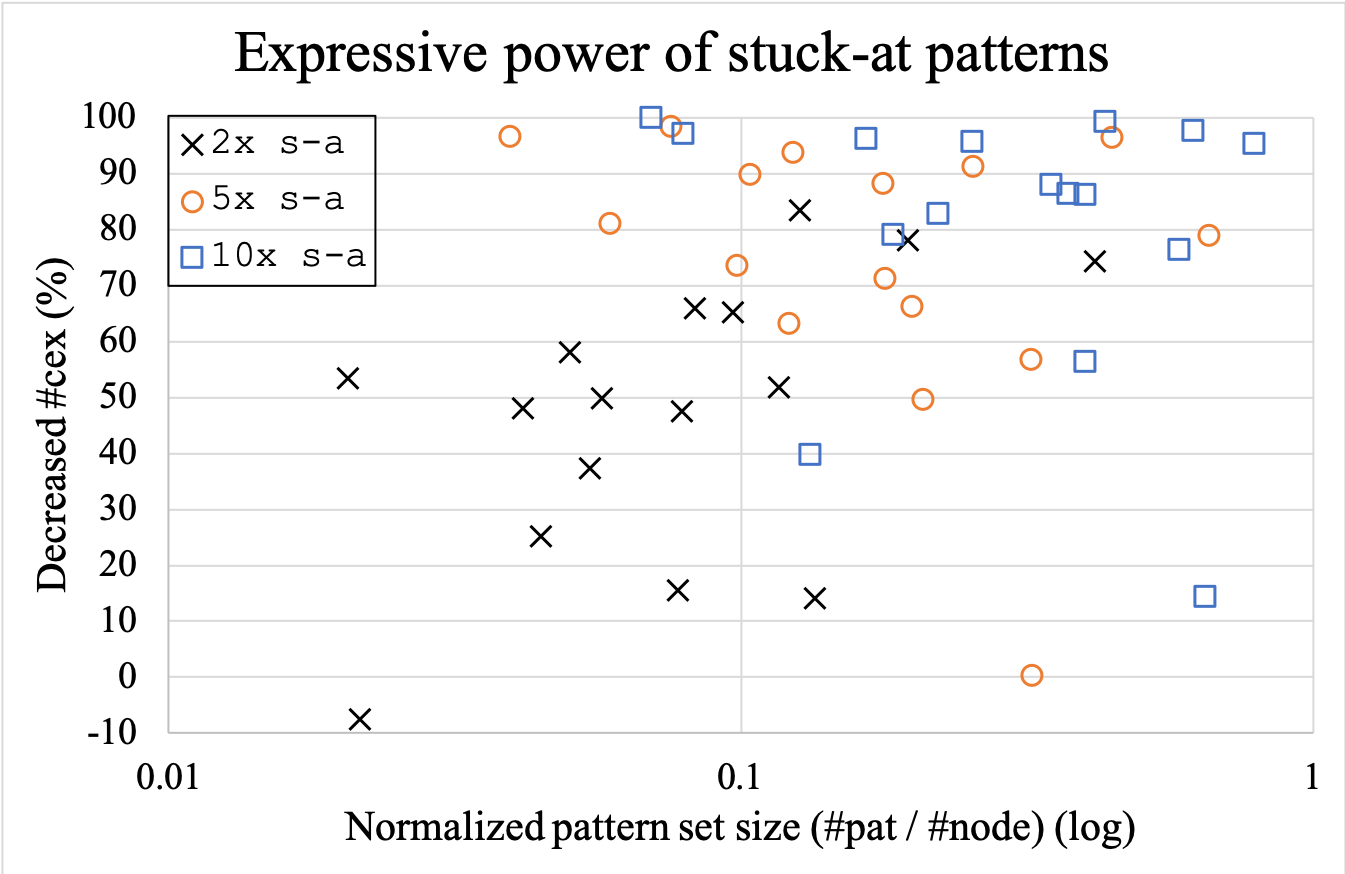}
\caption{Decreased percentages of counter-examples when using different sets of stuck-at simulation patterns, compared to using the ``\texttt{1x s-a}'' set.}\label{fig:exp-sa}
\end{figure}

\subsection{Pattern Generation Strategies}
In this section,  the expressive power of simulation patterns generated by $StuckAtCheck$ is compared with the case when observability is used (suffix ``\texttt{-obs}'') and/or when an initial random pattern set of size $256$ is used (prefix ``\texttt{rand 256}''). 

The observability check and observable pattern generation are done with a fan-out depth of $5$ levels. A set of $256$ random patterns is used as the baseline in Figure~\ref{fig:exp-strategies}. Note that there are a few benchmarks, for which the random pattern sets are more expressive than ``\texttt{1x s-a}'' and/or ``\texttt{1x s-a-obs}''. These are not shown in the figure. The geometric means of the sizes of the pattern sets are $143$ for ``\texttt{1x s-a}'', $244$ for ``\texttt{1x s-a-obs}'', $354$ for ``\texttt{rand 256 + 1x s-a}'' and $462$ for ``\texttt{rand 256 + 1x s-a-obs}''. On the other hand, the geometric means of the decreased percentages of the counter-examples are $91.3\%$, $96.5\%$, $97.1\%$ and $99.5\%$, respectively.

It is observed that patterns generated by $StuckAtCheck$ are usually more expressive than random patterns, except for a few, typically small, benchmarks. Also, using observability increases the expressive power of the generated patterns. Finally, seeding the pattern generation engine with an initial set of random patterns not only speeds up the generation process, but also makes the resulting patterns more expressive.

\begin{figure}[ht]
\centering
\includegraphics[width=0.48\textwidth]{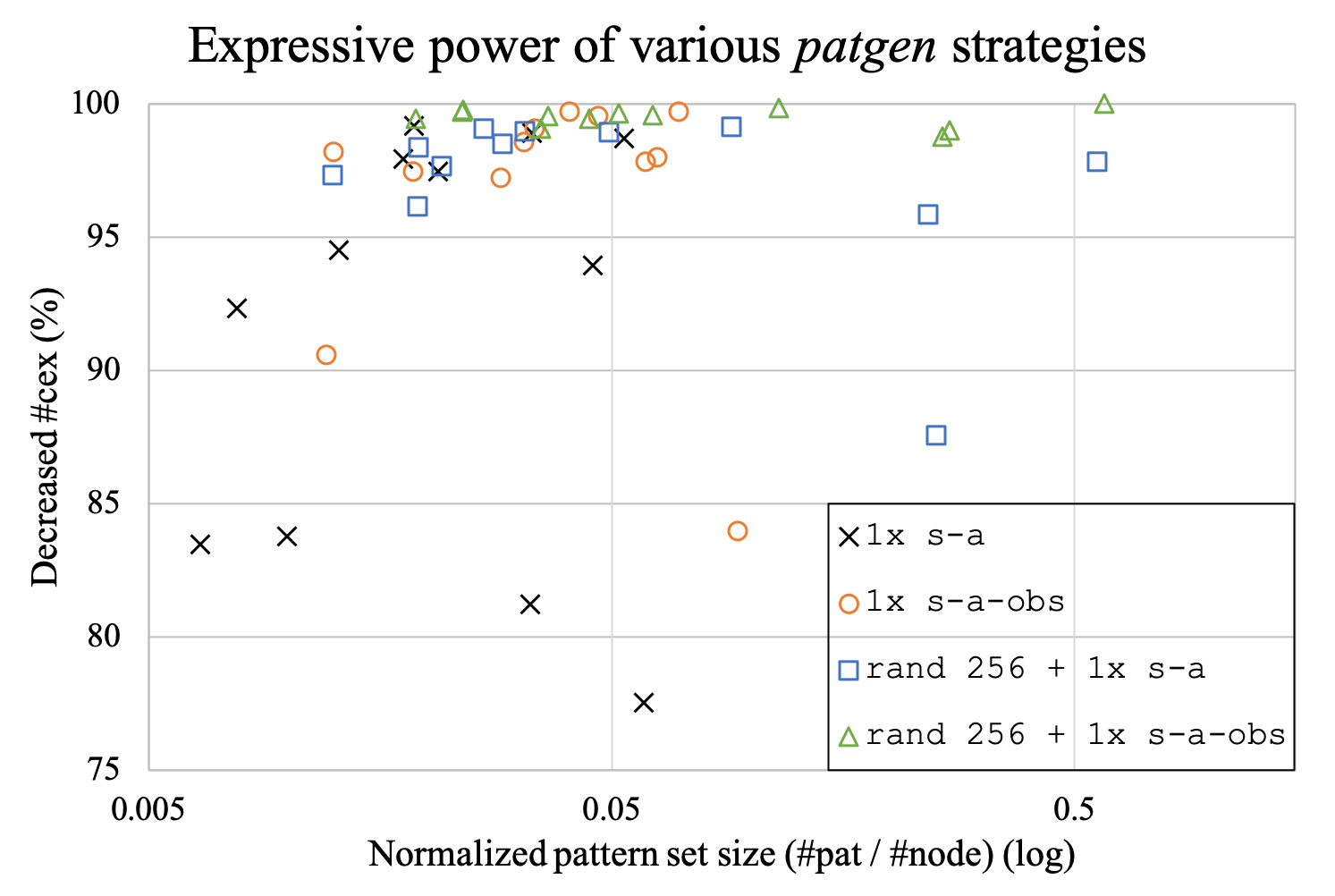}
\caption{The decrease of percentage of counter-examples when different patterns are used, relative to ``\texttt{rand 256}''. Four benchmarks are excluded because their baseline is more expressive than ``\texttt{1x s-a}'' and/or ``\texttt{1x s-a-obs}''.}\label{fig:exp-strategies}
\end{figure}

\subsection{Effect of Expressive Patterns in Resubstitution}\label{subsec:cex}
As stated in the introduction, an expressive set of simulation patterns is used to shift the computation effort from the optimization algorithms to  pattern pre-computation. Table~\ref{tbl:runtime} shows how the quality of the patterns affects the runtime of pattern generation (\textit{patgen}) and resubstitution (\textit{resub}). A better set of patterns (Table~\ref{tbl:runtime}, ``\texttt{5x s-a}'') efficiently filters out many illegal resubstitutions without calling the SAT solver, resulting in the reduced counter-example counts (\#cex) and faster runtimes.

Furthermore, in practice, when the same design is repeatedly synthesized during development or when simulation patterns are reused by different optimization engines, counter-examples from the previous runs can be saved for later use. In this case, additional counter-example count, generated during later runs, can go down to nearly zero, and the runtime is only spent on logic synthesis or verification tasks, such as proving equivalences among the nodes or computing dependency functions and validating them. The latter scheme will be used in the next section.

\begin{table}[htb]
  \caption{Resubstitution runtime as a function of the number of counter-examples produced.}\label{tbl:runtime}
  \scriptsize
  \centering
  \vspace{-1em}
  \begin{tabular}{|l||r|r|r||r|r|r|}
    \hline
    \hline           & \multicolumn{3}{c||}{\texttt{rand 256}} & \multicolumn{3}{c|}{\texttt{5x s-a}} \\
    \hline benchmark & \#cex &\multicolumn{2}{c||}{runtime (s)} & \#cex & \multicolumn{2}{c|}{runtime (s)} \\
    \hline  &  & \textit{patgen} & \textit{resub} &  & \textit{patgen} & \textit{resub} \\
    \hline
    ac97\_ctrl    &  72   & 0.01 & 0.53  & 47  & 3.16  & 0.43 \\
    aes\_core     &  50   & 0.02 & 2.22  & 9   & 3.91  & 1.95 \\
    DMA           &  1366 & 0.02 & 15.30 & 69  & 19.43 & 1.30   \\
    mem\_ctrl     &  2737 & 0.01 & 16.23 & 211 & 5.30  & 0.81   \\
    pci\_bridge32 &  1223 & 0.02 & 12.94 & 174 & 14.73 & 2.58  \\
    systemcaes    &  130  & 0.01 & 0.69  & 64  & 1.93  & 0.43   \\
    usb\_funct    &  1189 & 0.01 & 7.84  & 195 & 5.49  & 1.70   \\
    tv80          &  748  & 0.01 & 2.75  & 192 & 3.43  & 1.04   \\
    \hline
    \hline
  \end{tabular}
\end{table}

\subsection{Quality of Simulation-Guided Resubstitution}

This section shows the improvements in terms of resubstitution quality. Table~\ref{tbl:quality} compare the proposed framework with command \texttt{resub}~\cite{IWLS2006} in ABC~\cite{ABC}, which performs truth-table-based resubstitution. Because computing simulation patterns in our framework results in detecting combinational equivalences~\cite{fraig}, for a fair comparison, the benchmarks are pre-processed by repeating the command \texttt{ifraig} in ABC until no more size reduction is observed. The quality of results is measured using the reduction in circuit size after optimization, presented in the \textit{gain} columns.  Simulation patterns used in our framework are initially generated with ``\texttt{rand 256 + 1x s-a-obs}'' and then supplemented with the counter-examples generated from the previous runs of the same experiment.

The maximum cut size $K$ used to collect divisors in the TFI of the root node can be set in both flows. Since \cite{IWLS2006} relies on computing truth tables in the window, $K\leq10$ is typically used as a reasonable trade-off between efficiency and quality. In contrast, windowing in our framework is applied only to avoid potential runtime blow-up for large benchmarks, and $K$ can be set to arbitrarily large values when longer runtime is acceptable. 

Table~\ref{tbl:quality} shows that our framework achieves a $40.96\%$ improvement on average using the same, small window size, and that we are able to extend the window size and achieve up to $73.82\%$ improvement without runtime overhead. 

\begin{table*}[htb]
  \caption{Quality comparison against \texttt{ABC} with different window sizes.}\label{tbl:quality}
  \scriptsize
  \centering
  \vspace{-1.5em}
  \begin{tabular}{|l|R{0.6cm}|R{0.6cm}||R{1cm}|R{1cm}||R{1.2cm}|R{1.2cm}|R{1.2cm}||R{1.2cm}|R{1.2cm}|R{1.2cm}|}
    \hline
    \hline \multicolumn{3}{|c||}{\texttt{abc> ifraig} until sat.} & \multicolumn{2}{C{2.3cm}||}{\texttt{abc> resub -K 10}} & \multicolumn{3}{C{3.9cm}||}{Ours, $K=10$} & \multicolumn{3}{C{3.9cm}|}{Ours, $K=100$} \\
    \hline benchmark     &   size & \#PIs & gain & time (s) & gain & time (s) & improv. (\%) & gain & time (s) & improv. (\%) \\
    \hline
    ac97\_ctrl    & 14199  & 4482 & 177 & 0.13 & 178 & 0.25 & 0.56  & 181  & 0.26 & 2.26 \\
    aes\_core     & 21441  & 1319 & 322 & 0.42 & 343 & 1.18 & 6.52  & 496  & 3.44 & 54.04 \\
    des\_area     & 4827   &  496 & 88  & 0.07 & 105 & 0.11 & 19.32 & 104  & 0.55 & 18.18  \\
    DMA           & 21992  & 5070 & 195 & 0.24 & 229 & 0.87 & 17.44 & 286  & 2.05 & 46.67 \\
    i2c           & 1120   &  275 & 48  & 0.01 & 57  & 0.02 & 18.75 & 86   & 0.03 & 79.17  \\
    mem\_ctrl     & 8822   & 2281 & 218 & 0.09 & 1219& 0.27 & 459.17& 1350 & 1.42 & 519.27  \\
    pci\_bridge32 & 22521  & 6880 & 176 & 0.43 & 194 & 0.80 & 10.23 & 267  & 1.42 & 51.70 \\
    sasc          & 770    &  250 & 5   & 0.01 & 5   & 0.01 & 0.00  & 5    & 0.01 & 0.00 \\
    simple\_spi   & 1034   &  280 & 18  & 0.01 & 17  & 0.01 & -5.56 & 23   & 0.01 & 27.78 \\
    spi           & 3762   &  505 & 81  & 0.06 & 84  & 0.08 & 3.70  & 89   & 0.44 & 9.88 \\
    ss\_pcm       & 405    &  193 & 1   & 0.00 & 1   & 0.00 & 0.00  & 1    & 0.00 & 0.00 \\
    systemcaes    & 12108  & 1600 & 36  & 0.11 & 48  & 0.19 & 33.33 & 55   & 0.66 & 52.78  \\
    systemcdes    & 2857   &  512 & 138 & 0.04 & 155 & 0.10 & 12.32 & 161  & 0.28 & 16.67 \\
    tv80          & 9093   &  732 & 221 & 0.14 & 260 & 0.34 & 17.65 & 451  & 3.05 & 104.07 \\
    usb\_funct    & 15278  & 3620 & 452 & 0.15 & 581 & 1.03 & 28.54 & 1199 & 1.49 & 165.27  \\
    usb\_phy      & 440    &  211 & 12  & 0.00 & 16  & 0.01 & 33.33 & 16   & 0.01 & 33.33 \\
    \hline \multicolumn{3}{|c||}{average} 
                               & & 0.12 & & 0.33 & 40.96 & & 0.95 & 73.82  \\
    \hline
    \hline
  \end{tabular}
\end{table*}

\section{Conclusions and Future Works}\label{sec:concl}
The paper presented (1) a simulation-based logic optimization framework, which separates the computation of expressive simulation patterns and their use to validate Boolean optimization choices; (2) methods to increase expressiveness of simulation patterns, resulting in reduced runtime due to fewer SAT calls; (3) improvements to the flexibility of resubstitution candidates, resulting in better optimization quality.

This work has been partially motivated by the success of approximate logic synthesis, when substantial reductions in the circuit size are achieved without expensive Boolean computations, at the cost of introducing some errors into logic functions. We hope that future work in the area of simulation-based methods will help improve speed and quality of both exact and approximate logic synthesis.

In particular, future work may include developing strategies to refine and enhance the generated simulation patterns further, metrics to evaluate and sort the patterns, and methods to compress the pattern set while preserving  expressiveness. While resubstitution guided by simulation signatures automatically accounts for satisfiability don't-cares, observability don't-cares can also be considered in the validation of resubstitution, resulting in better quality. 

The search space of resubstitution candidates can be extended to include complex dependency functions requiring more than one gate. As shown in Section~\ref{subsec:cex}, using expressive patterns reduces the chance of encountering counter-examples, making it possible to further reduce the use of SAT solving by validating several resubstitutions at the same time, if almost all of them are legal. To deal with benchmarks containing millions of nodes, incremental CNF construction can be used to limit the size of the SAT instance because too many unrelated clauses slow down SAT solving. Finally, the framework can be extended to optimize mapped networks and to perform other types of Boolean optimization. 

\begin{acks}
This work was supported in part by the EPFL Open Science Fund, and by SRC Contract 2867.001, "Deep integration of computation engines for scalability in synthesis and verification".
\end{acks}

\end{document}